\documentclass[submission,copyright,creativecommons]{eptcs}
\providecommand{\event}{AFL 2014}
\usepackage{amssymb}
\usepackage{breakurl}
\usepackage[utf8]{inputenc}

\newcommand{\N}{{\mathbb N}}
\newcommand{\Z}{{\mathbb Z}}

\newcommand{\mod}{{\rm \,mod\,}}

\newtheorem{definition}{Definition}
\newtheorem{lemma}{Lemma}
\newtheorem{theorem}{Theorem}

\newenvironment{proof}[1][Proof]{\begin{trivlist}
\item[\hskip \labelsep {\bfseries #1}]}{\end{trivlist}}

\usepackage{url}

\title{Simplifying Nondeterministic Finite Cover Automata}

\author{Cezar C\^ampeanu
\institute{Department of Computer Science and Information Technology,\\
The University of Prince Edward Island, Canada}
\email{\quad ccampeanu@upei.ca}
}
\def\titlerunning{Simplifying NFCA}
\def\authorrunning{Cezar C\^ampeanu}
\begin{document}
\sloppy

\maketitle
\begin{abstract}
The concept of Deterministic Finite Cover Automata (DFCA) was introduced at WIA '98, 
as a more compact representation than Deterministic Finite Automata (DFA) for finite languages.
In some cases representing a finite language, Nondeterministic Finite Automata (NFA) may 
significantly reduce  the number of states used. 
The combined power of the succinctness of the representation of finite languages using 
both cover languages and non-determinism has been suggested, but never  systematically studied. 
In the present paper, for  nondeterministic finite cover automata (NFCA) and $l$-nondeterministic 
finite cover automaton ($l$-NFCA), we show that minimization can be as hard as minimizing NFAs for regular languages, 
even in the case of NFCAs using unary alphabets. Moreover, we show how we can adapt the 
methods used to reduce, or minimize the size of NFAs/DFCAs/$l$-DFCAs, for 
simplifying NFCAs/$l$-NFCAs.
\end{abstract}
\section{Introduction}

The race to find more compact representation for finite languages was started in 1959, 
when Michael O. Rabin and Dana Scott introduced the notion of Nondeterministic Finite Automata, 
and showed that the equivalent Deterministic Finite Automaton can be, 
in terms of number of states,  exponential larger than the NFA.
Since, it was proved in \cite{moore} that  we can obtain a polynomial algorithm for minimizing
DFAs, and in \cite{hopcroft} was proved that an $O(n\log n)$ algorithm exists.
In the meantime, several heuristic approaches have been proposed to reduce 
the size of NFAs \cite{heuristic,ilieYunfa}, but it was proved by 
Jiang and  Ravikumar \cite{ravikumar} that  NFA minimization problems are
hard; even in case of regular languages over a one letter alphabet, 
the minimization is NP-complete \cite{GruberHolzerNFAHard,ravikumar}.

On the other hand, in case of finite languages, we can obtain minimizing algorithms \cite{maslov,revuz} 
that are in the order of $O(n)$, where $n$ is the number of states of the original DFA.
In \cite{gapIJFCS,CoverAutomata,KornerGoeman} it has been shown that 
using Deterministic Finite Cover Automata to represent finite languages,
we have minimization algorithms as efficient as the best known algorithm for minimizing DFAs 
for regular languages. 

The study of the state complexity of operations on regular languages was initiated 
by Maslov in 1970 \cite{maslov,maslov1}, but has not become 
a subject of systematic study until 1992 \cite{KaiShengsc92}.
The special case of state complexity of operations 
on finite languages was studied in \cite{finiteop}.

Nondeterministic state complexity of regular languages 
was also subject of interest, for example in 
\cite{holzerKutribNFA,holzerKutribUnary,holzerKutribLata09,holzerKutribNFA09}.
To find lower bounds for the nondeterministic 
state complexity of regular languages, the fooling set technique,
or the extended fooling set technique may 
be used \cite{birget,Shallit,GruberHolzerNFAHard}.

In this paper we show that NFCA state complexity 
for a finite language $L$ can be exponentially lower 
than NFA or DFCA state complexity of the same language.
We modify the fooling set technique for cover automata, 
to help us prove lower bounds for NFCA state complexity
in section \ref{slowerbounds}. We also show that the (extended) fooling set technique 
is not optimal, as we have minimal NFCAs with arbitrary number of states, and the largest fooling set 
has constant size.
In section \ref{shard} we show that minimizing NFCAs is hard, and in section~\ref{sheuristic}
we show that heuristic approaches for minimizing DFAs or NFAs need a special treatment
 when applied to NFCAs, as many results valid for the DFCAs are no longer true for NFCAs.
In section \ref{openproblems}, we formulate a few open problems and future research directions.

\section{Notations and definitions}

The number of elements of a set $T$ is denoted by $\#T$. 
In case $\Sigma$ is an alphabet, i.e, finite non-empty set, the free monoid generated 
by $\Sigma$ is $\Sigma^*$, and it is the set of all words over $\Sigma$.
The length of a word $w=w_1w_2\ldots w_n$, $n\geq 0$, $w_i\in \Sigma$, 
$1\leq i\leq n$, is $|w|=n$.
The set of words of length equal to $l$ is $\Sigma^l$, the set of words of length less than or equal 
to $l$ is denoted by $\Sigma^{\leq l}$. 
In a similar fashion, we define $\Sigma^{\geq l}$, $\Sigma^{< l}$, or $\Sigma^{> l}$.
A finite automaton is a structure  $A=(Q,\Sigma,\delta,q_0,F)$, where
$Q$ is a finite non-empty set called the set of states, $\Sigma$ is an alphabet, $q_0\in Q$,  
$F\subseteq Q$ is the set of final states, and $\delta$ is the transition function.
For delta, we distinguish the following cases:
\begin{itemize}
 \item if $\delta:Q\times\Sigma\stackrel{\circ}{\longrightarrow} Q$, the automaton is deterministic; 
in case $\delta$ is always defined, the automaton is complete, otherwise it is incomplete;
 \item if $\delta:Q\times\Sigma\longrightarrow 2^Q$, the automaton is non-deterministic.
\end{itemize}
The language accepted by an automaton is defined by:
$L(A)=\{w\in \Sigma^*\mid \delta(\{q_0\},w)\cap F\neq \emptyset\}$, where
$\delta(S,w)$ is defined as follows:
$$\delta(S,\varepsilon)=S,$$
$$\delta(S,wa)=\displaystyle\bigcup_{q\in \delta(S,w)}\delta(\{q\},a).$$
Of course, $\delta(\{q\},a)=\{\delta(q,a)\}$ in case the automaton is deterministic, and 
$\delta(\{q\},a)=\delta(q,a)$, in case the automaton is non-deterministic.
\begin{definition}
Let $L$ be a finite language, and $l$ be the length of the longest word $w$ in $L$, i.e.,
$l=\max\{|w| \mid w\in L\}$\footnote{ We use the convention that $max\emptyset=0$.}.
If $L$ is a finite language, $L'$ is a cover language for $L$ if $L'\cap \Sigma^{\leq l}=L$.

A cover automaton for a finite language $L$ is an automaton that
recognizes a cover language, $L'$, for $L$.
An $l$-NFCA $A$ is a cover automaton for the language $L(A)\cap \Sigma^{\leq l}$.
\end{definition}

One could plainly see  that any automaton that recognizes $L$ is also a cover automaton.

The level of a state $s\in Q$ in a cover automaton $A=(Q,\Sigma,\delta,q_0,F)$ is
the length of the shortest word that can reach the state $s$, i.e., 
$level_A(s)=min\{|w|\mid s\in\delta(q_0,w)\}$.

Let us denote by $x_A(s)$ the smallest word $w$, according to quasi-lexicographical order, such that
$s\in \delta(q_0,w)$, see \cite{CoverAutomata} for a similar definition in case of DFCA.
Obviously,  $level_A(s)=|x_A(s)|$.

For a regular language $L$, $\equiv_L$ denotes the Myhil-Nerode equivalence of words \cite{hopcroft79}.

The similarity relation induced by a finite language $L$ is defined as follows\cite{CoverAutomata}:
$x\sim_L y$, if
for all $w\in \Sigma^{\leq l-\max\{|x|,|y|\}}$, $xw\in L$ iff $yw\in L$.
A dissimilar sequence for a finite language $L$ 
is a sequence $x_1,\ldots, x_n$ such that $x_i\not\sim_L x_j$, for all $1\leq i,j\leq n$
and
 $i\neq j$.

Now, we need to define the similarity for states in an NFCA, 
since it was the main notion used for DFCA minimization.

\begin{definition}
In an NFCA $A=(Q,\Sigma,\delta,q_0,F)$, two states $p,q\in Q$ are similar, written $s\sim_A q$,  if
$\delta(p,w)\cap F\neq \emptyset$ iff $\delta(q,w)\cap F\neq \emptyset$, 
for all $w\in \Sigma^{\leq l-\max\{level(p),level(q)\}}$.
\end{definition}

In case the NFCA $A$ is understood, we may omit the subscript $A$, i.e., we 
write $p\sim q$ instead of $p\sim_Aq$, also we can write $level(p)$ instead of $level_A(p)$.

We consider only non-trivial NFCAs for $L$, i.e., 
 NFCAs such that $level(p)\leq l$ for all states $p$. In case
$level(p)>l$, $p$ can be eliminated, and the resulting NFA is still a NFCA for $L$.
In this case, if $p\sim q$, then  either $p,q\in F$, or $p,q\in Q\setminus F$, because $|\varepsilon|\leq 
l-\max\{level(p),level(q)\}$.

Deterministic state complexity of a regular language $L$ is defined as the number of states of the minimal 
deterministic automaton recognizing $L$, and it is denoted by $sc(L)$:
$$
sc(L)=\min\{\#Q\mid A=(Q,\Sigma,\delta, q_o,F),\mbox{ deterministic, complete, and }L=L(A)\}.
$$ 
Non-deterministic state complexity of a regular language $L$ is defined as the number of states of the minimal 
non-deterministic automaton recognizing $L$, and it is denoted by $nsc(L)$:
$$
nsc(L)=\min\{\#Q\mid A=(Q,\Sigma,\delta, q_o,F), \mbox{ non-deterministic  and }L=L(A)\}.
$$ 
For finite  languages $L$, we can also define deterministic cover state complexity $csc(L)$
and non-deterministic cover state complexity $ncsc(L)$:
\begin{eqnarray*}
csc(L) & = & \min\{\#Q\mid A=(Q,\Sigma,\delta, q_o,F), \mbox{ deterministic, complete, and }\\
       &   & L=L(A)\cap \Sigma^{\leq l}\},\\
ncsc(L)& = & \min\{\#Q\mid A=(Q,\Sigma,\delta, q_o,F),\mbox{ non-deterministic,  and }\\
       &   & L=L(A)\cap \Sigma^{\leq l}\}. 
\end{eqnarray*}

Obviously,  $ncsc(L)\leq nsc(L)\leq sc(L)$, but also 
$ncsc(L)\leq csc(L)\leq sc(L)$.
Thus, non-deterministic finite cover automata can be considered to be one of 
the most compact representation of finite languages.

\section{Lower-bounds and Compression Ratio 
for NFCAs}
\label{slowerbounds}

We start this section analyzing few examples where nondeterminism, 
or the use of cover language, reduce the state complexity.
Let us first analyze the type of languages where non-determinism, combined with cover properties,
 reduce significantly the state complexity.

We choose the language 
 $L_{F_{m,n}}=\{a,b\}^{\leq m}a\{a,b\}^{n-2}$, where $m,n\in \N$.
In Figure~\ref{f1} we can see an NFA recognizing $L$ with $m+n$ states.
We must note that the longest word in the language 
has $m+n-1$ letters.
Let us analyze if the automaton in Figure~\ref{f1}  is minimal.
The fooling set technique, introduced in \cite{Chrobak} and \cite{gramlich}, and used to prove the lower-bound 
for state complexity of NFAs,  is  stated in \cite{birget,Chrobak} as follows:

\begin{lemma}
\label{lfst1}
Let $L\subseteq \Sigma^*$ be a regular language, and suppose there  exists a set of pairs 
$S=\{(x_i,y_i)\mid 1\leq i\leq n\}$, with the following properties:
\begin{enumerate}
 \item 
  \label{sfst}
  If $x_i y_i\in L$, for $1\leq i\leq n$ and $x_iy_j\notin L$, for all 
$1\leq i,j\leq n$, $i\neq j$, then $nsc(L)\geq n$. The set $S$ is called {\em a fooling set} for $L$.
 \item 
 \label{extfst}
If $x_i y_i\in L$, for $1\leq i\leq n$ and for $1\leq i,j\leq n$, if $i\neq j$, implies  
either $x_iy_j\notin L$ 
or $x_jy_i\notin L$, 
 then $nsc(L)\geq n$. The set $S$ is called {\em an extended fooling set} for $L$.
\end{enumerate}
\end{lemma}

Now consider the following set of pairs of words:
$S=\{(b^mab^j,b^{n-2-j})\mid 0\leq j\leq n-2\}\cup\{(a^i,b^{m-i}ab^{n-2})\mid 0\leq i\leq m\}=
\{(x_k,y_k)\mid 1\leq k\leq m+n\}$.

For $(x_k,y_k)\in S$, we have  that
\begin{enumerate}
 \item $x_ky_k=b^mab^jb^{n-2-j}=b^mab^{n-2}\in L$, or
 \item $x_ky_k=a^ib^{m-i}ab^{n-2}\in L$.
\end{enumerate}

Let us examine for each $1\leq k,h\leq m+n$, $k\neq h$ if  the words  
$x_ky_h$ and 
$x_hy_k$ are also in $L$. We have the following possibilities:
\begin{enumerate}
  \item Case I
  \begin{enumerate}
   \item $x_ky_h=b^mab^ib^{n-2-j}\notin L$, and
   \item $x_hy_k=b^mab^jb^{n-2-i}\notin L$.
  \end{enumerate}
 \item Case II
  \begin{enumerate}
   \item $x_ky_h=a^ib^{m-j}ab^{n-2}\in L$, if $i<j$, but 
   \item $x_hy_k=a^jb^{m-i}ab^{n-2}\notin L$, if $i<j$ (because $|a^jb^{m-i}ab^{n-2}|=m+n-1+j-i>m+n-1$).
  \end{enumerate}
 \item Case III
  \begin{enumerate}
   \item $x_ky_h=b^{m}ab^jb^{m-i}ab^{n-2}\notin L$ (because $|b^{m}ab^jb^{m-i}ab^{n-2}|=m+1+j+m+1+n-2>m+n-1$), or
   \item $x_hy_k=a^ib^{n-2-j}\in L$ if $n-2-j+1+i>n$, 
         or 
         $x_hy_k=a^ib^{n-2-j}\notin L$ $n-2-j+1+i<n$.
  \end{enumerate}
\end{enumerate}
From the statement \ref{extfst}. of Lemma~\ref{lfst1}, it follows that the NFA is minimal. 
We must note the following:
\begin{enumerate}
 \item we cannot use the weak form \ref{sfst} to prove the lower-bound;
 \item when proving the lower-bound, we concatenate words to obtain a word 
of length greater than the maximum length of the words in the language, and 
that's why $x_iy_j$ is rejected.
Since in case of cover automata such words will be automatically rejected, there is no doubt that any fooling set type 
technique we may use to prove the lower-bound for NFCAs must consider the length, and 
we should ignore the cases when the length exceeds the maximal one.
\end{enumerate}

Hence, the fooling set technique introduced in \cite{Chrobak} and \cite{gramlich}, and used to prove the lower-bound 
for state complexity of NFAs, can be modified to prove 
a lower-bound for minimal NFCAs, and it can be formulated  for cover languages
as an adaptation of Theorem~1 in \cite{GruberHolzerNFAHard}.
\begin{lemma}
\label{lcfst}
Let $L\subseteq \Sigma^{\leq l}$ be a finite language such that the longest word in $L$ has the length $l$, 
and suppose there  exists a set of pairs $S=\{x_i,y_i)\mid 1\leq i\leq n\}$, with the following properties:
\begin{enumerate}
 \item 
  \label{scfst}
  If $x_i y_i\in L$ for $1\leq i\leq n$ and 
for $1\leq i,j\leq n$, $i\neq j$, and $x_iy_j\in \Sigma^{\leq l}$, we have that 
$x_iy_j\notin L$, then $ncsc(L)\geq n$. 

The set $S$ is called {\em a fooling set} for $L$.
 \item 
 \label{cfstext}
If $x_i y_i\in L$, for $1\leq i\leq n$ and for $1\leq i,j\leq n$, if $i\neq j$, implies 
either $x_iy_j\in \Sigma^{\leq l}$ and $x_iy_j\notin L$, 
or $x_jy_i\in \Sigma^{\leq l}$ and $x_jy_i\notin L$ for all, 
 then $ncsc(L)\geq n$. 

The set $S$ is called {\em an extended fooling set} for $L$.
\end{enumerate}
\end{lemma}
\begin{proof}
Assume there exists an NFCA $A=(Q,\Sigma,\delta,q_0,F)$, with $m$ states accepting $L$.
For each $i$, $1\leq i\leq n$, $x_iy_i \in L$, therefore we must have a state $s_i \in\delta(q_0,x_i)$ 
and $\delta(s_i, y_i) \cap F\neq \emptyset$.
In other words, there exists a state $f_i \in F$ and $f_i \in \delta(s_i, y_i)$. 
\begin{enumerate}
 \item 
We claim  $s_i\notin \delta(q_0,x_j)$ for all $j \neq i$.
If $s_i \in \delta(q_0,x_j)$,
then $f_i \in \delta(s_i,y_i) \subseteq \delta(q_0,x_jy_i)$, and because $|x_jy_i|\leq l$, it follows that 
  $x_jy_i \in  L$, a contradiction.
\item 
We consider the function $f:\{1,\ldots,n\}\longrightarrow Q$ 
defined by $f(i)=s_i$, $s_i$ as above.
We claim that $f$ is injective.
If $f(i)=f(j)$, then
$\delta(f(i),y_i)=\delta(f(j),y_i)$,
also 
$\delta(f(j),y_j)=\delta(f(i),y_j)$.
Because $\delta(f(i),y_i)\cap F\neq \emptyset$, we also have that $\delta(f(j),y_i)\cap F\neq \emptyset$,
and because $|x_iy_j|\leq l$, it follows that $x_iy_j\in L$, a contradiction.
If $|x_jy_i|\leq l$, using the same reasoning, will follow that $x_jy_i\in L$.
In both cases we have a contradiction, thus $Q$ must have at least $n$ elements.$\Box$
\end{enumerate}
\end{proof}

For the example above, we discover that we cannot have more than one pair 
of the form $(a^i,b^{m-i}ab^{n-2})$, thus, applying the extended fooling set technique for NFCAs,
 the minimum number of states in a minimal 
NFCA is at least $n-2+1+1=n$. 
This proves that the NFCA presented in Figure~\ref{f2} is minimal.

It is easy to check that any two distinct words $w_1,w_2\in \Sigma^{\leq n-1}$, $w_1\neq w_2$, 
are not similar with respect to $\sim_L$. It follows that for the language presented in 
Figure~\ref{f1}, $csc(L)\geq 2^{n-1}$. One can also verify that for two distinct words
$uay$ and $wax$,
if  $|y|\neq |x|$, $|x|,|y|\leq n-2$, they are distinguishable;
also, in case 
$|x|=|y|\leq n-2$, the word $a^{n-2-|x|}$ will distinguish between all the words for which
$|u|<n-2-|x|$ or $|w|<n-2-|x|$, thus the number of states in the minimal DFA is even
 larger than  $csc(L)$.
In case $m=2$  and $n=4$, the minimal DFCA is presented in Figure~\ref{f3}. A simple 
computation shows us that the corresponding minimal DFA has 15 states.

\begin{figure}[h]
\begin{center}
\begin{picture}(200,70)
\put(10,20){\oval(15,15)}
\put(7,17){$0$}

\put(10,50){\oval(18,15)}
\put(3,47){$-1$}
\put(44,50){\vector(-1,0){24}}
\put(25,55){$a,b$}
\put(10,43){\vector(0,-1){15}}
\put(-10,30){$a,b$}
\put(14,43){\vector(2,-1){36}}
\put(35,35){$a$}

\put(18,20){\vector(1,0){30}}
\put(25,25){$a$}
\put(56,20){\oval(15,15)}
\put(54,17){$1$}    

\put(54,50){\oval(18,15)}
\put(46,47){$-2$}    
\put(94,50){\vector(-1,0){30}}
\put(75,55){$a,b$}
\put(54,43){\vector(0,-1){15}}
\put(56,35){$a$}

\put(164,50){\oval(22,15)}
\put(156,47){$-m$}    
\put(185,50){\vector(-1,0){10}}
\put(152,50){\vector(-1,0){30}}
\put(125,55){$a,b$}
\put(152,50){\vector(-4,-1){92}}
\put(90,38){$a$}

\put(64,20){\vector(1,0){30}}
\put(76,25){$a,b$}
\put(102,20){\oval(15,15)} 
\put(98,17){$2$} 
\put(109,20){\vector(1,0){20}}
\put(112,25){$a,b$}
\put(145,20){\vector(1,0){20}}
\put(145,25){$a,b$}
\put(180,20){\oval(30,15)}
\put(170,17){$n-2$}
\put(195,20){\vector(1,0){25}}
\put(200,25){$a,b$}
\put(235,20){\oval(25,15)}
\put(224,17){$n-1$}    
\put(235,20){\oval(30,18)}  
\end{picture}
\end{center}
\caption{An NFA with $m+n$ states for the language $L_{F_{m,n}}=\{a,b\}^{\leq m}a\{a,b\}^{n-2}$.\newline
}
\label{f1}
\end{figure}
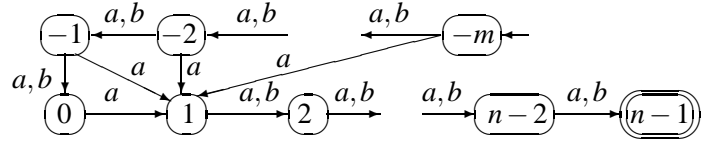

\begin{figure}[h]
\begin{center}
\begin{picture}(200,70)
\put(-5.5,30){\vector(1,0){5}}
\put(7,30){\oval(15,15)}
\put(5,27){0}
\put(7,49){\oval(15,15)[t]}
\put(-0.5,49){\line(1,-3){4}}
\put(14.5,49){\vector(-1,-3){4}}
\put(12,58){$a,b$}
\put(15,30){\vector(1,0){30}}
\put(25,35){$a$}
\put(52,30){\oval(15,15)}
\put(50,27){1}    
\put(60,30){\vector(1,0){30}}
\put(65,35){$a,b$}
\put(97,30){\oval(15,15)} 
\put(95,27){2} 
\put(105,30){\vector(1,0){20}}
\put(110,35){$a,b$}
\put(145,30){\vector(1,0){20}}
\put(145,35){$a,b$}
\put(180,30){\oval(30,15)}
\put(170,27){$n-2$}
\put(195,30){\vector(1,0){25}}
\put(200,35){$a,b$}
\put(235,30){\oval(29,15)}
\put(225,27){$n-1$}    
\put(235,30){\oval(32,18)}  
\end{picture}
\end{center}
\caption{An NFCA with $n$ states for the language $L_{F_{m,n}}=\{a,b\}^{\leq m}a\{a,b\}^{n-2}$, that is the same as 
the one in Figure~\ref{f1}. 
In case $m=2$ and $n=4$, the language is the same as the one described in Figure~\ref{f3}.\newline
An equivalent minimal NFA has $m+n$ states.
}
\label{f2}
\end{figure}
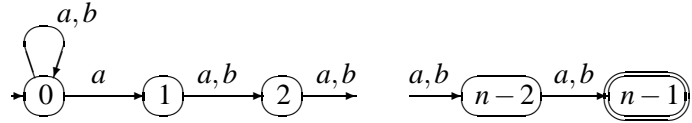

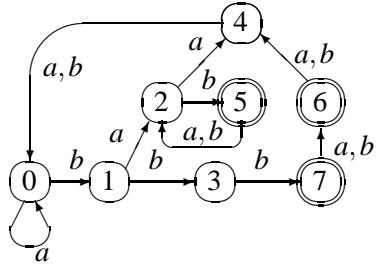
\begin{figure}[h]
\begin{center}
\begin{picture}(200,100)
\put(10,30){\oval(15,15)}
\put(7,27){0}
\put(18,30){\vector(1,0){15}}
\put(25,35){$b$}

\put(10,13){\oval(15,15)[b]}
\put(17,13){\vector(-1,2){5}}
\put(3,13){\line(1,2){5}}
\put(12,0){$a$}

\put(40,30){\oval(15,15)}
\put(37,27){1}    

\put(46,35){\vector(1,2){9}}
\put(40,45){$a$}

\put(60,60){\oval(15,15)} 
\put(57,57){2} 

\put(68,60){\vector(1,0){15}}
\put(75,65){$b$}

\put(75,48){\oval(30,10)[b]}
\put(60,48){\vector(0,1){5}}
\put(90,48){\line(0,1){5}}
\put(68,46){$a,b$}

\put(66,66){\vector(1,1){18}}
\put(70,80){$a$}

\put(90,90){\oval(15,15)} 
\put(87,87){4} 
\put(90,90){\oval(15,15)} 

\put(83,90){\line(-1,0){50}}
\put(50,50){\oval(80,80)[tl]}
\put(10,50){\vector(0,-1){12}}
\put(15,70){$a,b$}

\put(90,60){\oval(15,15)}
\put(90,60){\oval(18,18)}  
\put(87,57){5}

\put(48,30){\vector(1,0){24}}
\put(55,35){$b$}

\put(80,30){\oval(15,15)}
\put(77,27){$3$}

\put(88,30){\vector(1,0){24}}
\put(95,35){$b$}

\put(120,30){\oval(15,15)}
\put(117,27){$7$}    
\put(120,30){\oval(18,18)}  

\put(120,40){\vector(0,1){10}}
\put(125,40){$a,b$}

\put(120,60){\oval(15,15)}
\put(117,57){$6$}    
\put(120,60){\oval(18,18)}  

\put(115,68){\vector(-1,1){18}}
\put(110,75){$a,b$}

\end{picture}
\end{center}
\caption{A minimal DFCA with $8$ states for the language $L_{F_{2,4}}=\{a,b\}^{\leq 2}a\{a,b\}^{2}$, $l=5$.\newline
The equivalent minimal DFA has $15$ states.}
\label{f3}
\end{figure}

This language example shows that NFCAs may be a much more compact representation 
for finite languages than NFAs, or even DFCAs, and motivates the study of such objects.
In terms of compression, clearly the number of states in the NFCA 
is exponentially smaller than the number of states in the DFA, and in some 
cases, even exponentially smaller than in an NFA.

\begin{figure}
 \begin{center}
\begin{picture}(200,50)
 \put(2,40){\vector(1,0){5}}
 \put(15,40){\oval(15,15)}
 \put(13,37){$0$}
 \put(23,40){\vector(1,0){24}}
 \put(30,45){$a$}
 \put(55,40){\oval(15,15)}
 \put(53,37){$1$}
 \put(55,28){\vector(0,1){5}}
 \put(63,40){\vector(1,0){20}}
 \put(70,45){$a$}
 \put(92,40){\oval(15,15)}
 \put(92,40){\oval(18,18)}
 \put(90,37){$2$}
 \put(101,40){\vector(1,0){20}}
 \put(105,45){$a$}
 \put(135,40){\vector(1,0){12}}
 \put(140,45){$a$} 
 \put(161,40){\oval(25,15)}
 \put(161,40){\oval(28,18)}
 \put(150,37){$k-1$}
 \put(175,40){\vector(1,0){25}}
 \put(180,45){$a$}

 \put(210,40){\oval(20,15)}
\put(210,40){\oval(22,18)}
 \put(207,37){$k$}
 \put(133,30){\oval(156,20)[b]}
 \put(133,22){$a$}
\put(211,27){\line(0,1){5}}
\end{picture}
\end{center}
\caption{An NFA/NFCA $A_k$ for $L_{l,k}$.} 
\label{f4}
\end{figure}
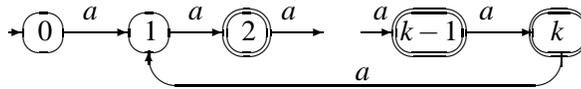

Let's set $\Sigma=\{a\}$, $l>k\geq 2$,
and choose the following language:
\begin{equation}
\label{Lprime}
L_{l,k}=a(\Sigma^{\leq l}-\{(a^k)^n\mid n\geq 0\}).  
\end{equation}

In Figure~\ref{f4},  the NFCA $A_{k}$ accepts the language $L_{l,k}$, therefore
 $ncsc(L_{l,k})\leq csc(L_{l,k})\leq sc(L_{l,k})\leq \min(l+1,k+1)=k+1$.
It is known \cite{Chrobak,holzerKutribUnary,pighizzini} 
that the automaton  $A_k$ is minimal NFA for $\displaystyle\bigcup_{l\in\N} L_{l,k}$, if $k$ is a prime number.
However, this may not be a minimal NFCA, as illustrated by the example in Figure~\ref{f5}, where $A_7$ 
is not minimal for $L_{9,7}$, even if it is minimal NFA for the cover language.

\begin{figure}
\begin{center}
\begin{picture}(150,75)
 \put(4,20){\vector(1,0){5}}
 \put(15,20){\oval(12,12)}
 \put(13,17){$0$}
 \put(22,20){\vector(1,0){25}}
 \put(30,25){$a$}
 \put(20,23){\vector(1,1){28}}
 \put(30,40){$a$}
 \put(54,52){\oval(12,12)}
 \put(52,50){$3$}
 \put(60,52){\vector(1,0){25}}
 \put(65,55){$a$}
 \put(92,52){\oval(12,12)}
 \put(92,52){\oval(15,15)}
 \put(90,50){$4$}
 \put(100,52){\vector(1,0){25}}
 \put(110,55){$a$}
 \put(132,52){\oval(12,12)}
 \put(132,52){\oval(15,15)}
 \put(130,50){$5$}
\put(92,60){\oval(80,15)[t]}
 \put(90,70){$a$}
\put(132,58){\line(0,1){5}}
\put(52,63){\vector(0,-1){5}}

 \put(55,20){\oval(12,12)}
 \put(53,17){$1$}
 \put(55,8){\vector(0,1){5}}
 \put(61,20){\vector(1,0){25}}
 \put(70,25){$a$}
 \put(92,20){\oval(12,12)}
 \put(92,20){\oval(15,15)}
 \put(90,17){$2$}

 \put(75,10){\oval(40,15)[b]}
 \put(75,5){$a$}
\put(95,9){\line(0,1){5}}
\end{picture}
\hspace*{2cm}
\begin{picture}(150,75)(20,0)
 \put(4,20){\vector(1,0){5}}
 \put(15,20){\oval(12,12)}
 \put(13,17){$0$}
 \put(22,20){\vector(1,0){20}}
 \put(25,25){$a$}
 \put(48,20){\oval(12,12)}
 \put(45,17){$1$}
 \put(48,8){\vector(0,1){5}}
 \put(54,20){\vector(1,0){20}}
 \put(60,25){$a$}
 \put(82,20){\oval(12,12)}
 \put(82,20){\oval(15,15)}
 \put(80,17){$2$}
 \put(88,20){\vector(1,0){20}}
 \put(105,25){$a$}
 \put(125,20){\vector(1,0){14}}
 \put(130,25){$a$} 
 \put(145,20){\oval(12,12)}
 \put(145,20){\oval(15,15)}
 \put(142,17){$6$}

 \put(97,10){\oval(98,15)[b]}
 \put(133,5){$a$}
\put(146,9){\line(0,1){5}}
\end{picture}
\end{center}
\caption{A minimal NFCA for $L_{9,7}$, left, and a  minimal NFA for a cover language, right.} 
\label{f5}
\end{figure}
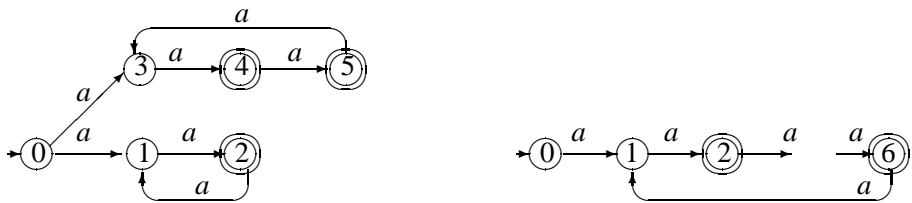

We apply the extended fooling set technique for the language $L_{l,k}$. 
Because the alphabet is unary, all the  words in an extended fooling
 set $S$  are powers of $a$: 
$S\supseteq\{(a^{i_1},a^{j_1}),(a^{i_2},a^{j_2}),(a^{i_3},a^{j_3}),\ldots , (a^{i_r},a^{j_r})\}$, 
for some $r\in \N$.
A simple computation shows that if $i_1,\ldots, i_r>1$, and $i_1+j_2 = z_{12}k+1$ and  $i_1+j_3 = z_{13}k+1$ 
for some $z_{12},z_{13}\in \N$, then $i_2+j_3\neq z_{23}k+1$ and 
$i_3+j_2\neq z_{32}k+1$, for any $z_{23},z_{32}\in\N$.
It follows that $r\leq 3$. 

Let $A$ be an NFA accepting $L\supseteq L_k$, and we can consider that it is already in Chrobak normal form, 
as it is ultimately periodic.
Thus, for each $L$, $nsc(L)\geq p_1+\ldots p_s$, where $p_i$ are primes, and each cycle has $p_i^{k_i}$ states, 
$1\leq i\leq s$. 
Now, let us prove that $A_k$ is minimal for some language $L_{l,k}$, $l\geq k$.

Assume there exists an automaton $B=(Q_B,\Sigma,\delta_B,q_{0,B},F_B)$ with $m$ states, $m\leq k+1$ such that
 $L(B)=L_{l,k}$.
It follows that the language $L(B)$ will contain words 
with a length $x+hy$ for $x,y\leq k$, and all $h\in \N$. For $h$ large enough, one of these words will be 
of length multiple of $k$ plus $1$, therefore, for large enough $l$, i.e., greater than some $l_0$,
$L_{l,k}\neq L(B)$. Thus, the number of states in $B$ is at least $k$.
$A_k$ is also a minimal NFCA for languages $L_{l,k}$, $l\geq l_0$, hence it follows that Theorem~7 in 
\cite{GruberHolzerNFAHard} is also valid for cover automata:

\begin{theorem}
 There is a sequence of languages $(L_{l,k})_{k\geq 2}$ such that the nondeterministic cover complexity
of $L_{l,k}$ is at least $k$, but the extended fooling set for $L_{l,k}$ is of size $c$,
 where $c$ is a constant.
\end{theorem}

Now, we are ready to check how hard is to obtain this minimal representation
of a finite language.

\section{Minimization Complexity}
\label{shard}

In this section we show that minimizing NFCAs is hard, and we'll show it with 
the exact same arguments from \cite{gruberholzerunary},
used to prove that minimizing NFAs is hard.
We will describe the construction from \cite{gramlich,gruberholzerunary}, showing 
that we can also use it with only a minor addition for cover NFAs.
To keep the paper self contained, we include a complete description,  
and emphasize the changes required for the cover automata, rather 
than just presenting the differences. 

Let us consider a logical formula $F \in 3SAT$, 
 in the conjunctive normal form, i.e., $F=\displaystyle\bigwedge_{i=1}^mC_i$,
where each clause $C_i$, $1\leq i\leq m$,
is defined using variables $x_1,\ldots,x_n$, $C_i=u_1\vee u_2\vee u_3$, and 
each $u_j$, $1\leq j\leq 3$ are either $x_i$ or $\neg x_i$.
Let $p_1,p_2,\ldots,p_n$ be distinct prime numbers such that
$p_1<p_2<\ldots<p_n$. We set $k=\prod_{i=1}^np_i$, and 
using Chinese Remainder Theorem \cite{china}\footnote{Theorem I.3.3, page 21},
it follows that the function $f:\Z_k\longrightarrow \prod_{i=1}^n\Z_{p_i}$ is bijective.
We need to define a language $L_F$ and a natural number $l$ such that
$L_F=\{a\}^*$, if and only if  $F$ is unsatisfiable, therefore, the finite language 
$L_F\cap \Sigma^{\leq l}$ has $\{a\}^*$ as a cover language.
We can construct an automaton $B_i$ in $O( p_n)$ in a similar fashion as we build automata $A_k$
that recognizes the language
$L(B_i)=\{a^n \mid n 
\mod 
p_i\notin\{0,1\}\}$.
Let $B$ be an automaton recognizing $\bigcup_{i=1}^nL(B_i)$. 
It is clear that it can be constructed in $O(n\cdot p_n)$ time.
For each clause $C_i$ such that $a_1,a_2,a_3$ is an assignment of its 
variables  for which 
$C_i$ is not satisfied, we define $L_{C_i}=\cap_{i=1}^3\{a^n\mid n \mod p_i =a_i\}$.
An automaton $C_i$ accepting  $L_{C_i}$ can be constructed 
in $O(p_n^3)$ time\footnote{Using Cartesian product construction, for example.}. 
Setting $L_F=\bigcup_{i=1}^m L_{C_i}\cup L(B)$, it follows that 
$L_F=\{a\}^*$ iff $F$ is satisfiable. Moreover, $L_F$ is a cyclic language with period at most $k$,
 thus setting $l=k$, we have that 
$L_F\cap \{a\}^{\leq l}$ has $\{a\}^*$
as a cover language iff
 $F$ is satisfiable. 
Since according to \cite{Primtest}, primality test can be done in polynomial time, 
 we can find the first $n$ prime numbers  in polynomial time, 
which means that our NFA construction can also be done in polynomial time.
If $F$ is unsatisfiable, then $ncsc(L)=1$, if $F$ is satisfiable,
then the  minimal period of $L_F$ is $\frac{l}{2}$, according to \cite{Chrobak,gramlich}, 
 and the minimal number of states 
in an NFA is at least equal to the largest prime number dividing its period, which is $p_n$. 
Using  the same argument as in \cite{gruberholzerunary}, it 
follows that the existence of a polynomial algorithm to decide if  $ncsc(L)=o(n)$ implies that 
$nsc(L)=o(n)$, therefore we can solve $3SAT$ in polynomial time, i.e., $P=NP$.
Consequently, we proved that
\begin{theorem}
 Minimizing either NFCAs or $l$-NFCAs is at least NP-hard. 
\end{theorem}

\section{Reducing the Number of States of NFCAs}
\label{sheuristic}

Assume the DFA $A=(Q,\Sigma,\delta,q_0,F)$ is minimal for $L$, and the minimal NFA is
$A'=(Q',\Sigma,\delta',q_0,F)$, 
where $Q'=Q-\{d\}$, $\delta'(s,p)=\delta(s,p)$, if $\delta(s,p\in Q')$ and $\delta'(s,p)=\emptyset$ if $\delta(s,p)=d$.
In other words, the minimal NFA is the same as the DFA, except that we delete the dead state.
We may have a minimal DFCA as $A$, and $A'$ as a minimal NFA, but not as a minimal NFCA,  as illustrated by
$A_7$ and $L_{9,7}$.

We need to investigate if classical methods to reduce the number of states in an NFA or DFA/DFCA
can also be applied to NFCAs, thus, we first analyze the state  merging technique.
For NFAs,  we distinguish between  two main ways of merging states: (1) a weak method,
where two states are merged by simply collapsing one into the other, and consolidate
all their input and output transitions, and (2), a strong method, where one state is
merged into another one by redirecting its input transitions toward the other state,
and completely deleting it and all its output transitions. 
The same methods are considered for NFCAs.

\begin{definition}
 Let $A=(Q,\Sigma,\delta,q_0,F)$ be a NFCA for the finite language $L$. 
\begin{enumerate}
 \item We say that the state $q$ is {\em weakly mergible} in state $p$
if the automaton
$A'=(Q',\Sigma,\delta',q_0,F')$, where
$Q'=Q-\{q\}$, $F'=F\cap Q'$, and 
$$
\delta(s,a)=\left\{\begin{array}{ll}
                      \delta(s,a),&\mbox{ if  }\delta(s,a)\subseteq Q'\mbox{ and }s\neq p,\\
                      (\delta(s,a)\setminus\{q\})\cup \{p\},&\mbox{ if  }q\in \delta(s,a)\mbox{ and }s\neq p,\\
                       (\delta(s,a)\cup\delta(q,a))\setminus\{q\},&\mbox{ if  }s=p\\
                     \end{array}
\right.
$$
 is also a NFCA for $L$.
In this case we write $p\precapprox q$.
\item We say that the state $q$ is {\em strongly mergible} in state $p$, 
if the automaton
$A'=(Q',\Sigma,\delta',q_0,F')$, where
$Q'=Q-\{q\}$, $F'=F\cap Q'$, and 
$$
\delta(s,a)=\left\{\begin{array}{ll}
                      \delta(s,a),&\mbox{ if  }\delta(s,a)\subseteq Q'\\
                      (\delta(s,a)\setminus\{q\})\cup \{p\},&\mbox{ if  }q\in \delta(s,a),
                     \end{array}
\right.
$$
 is also a NFCA for $L$.
In this case we write $p\precsim q$.
\end{enumerate}
\end{definition}
In case $p\precapprox q$, 
$(L^L_pL^R_p\cup L^L_pL^R_q\cup L^L_qL^R_p\cup L^L_qL^R_q) \cap \Sigma^{\leq l}\subseteq L$
and in case $p\precsim q$, 
$L^L_qL^R_q\cap \Sigma^{\leq l}\subseteq 
(L^L_pL^R_p\cup L^L_qL^R_p) \cap \Sigma^{\leq l} \subseteq L$, where for $s\in Q$
$L^L_s=\{w\in \Sigma^*\mid s\in \delta(q_0,w)\}$ and
$L^R_s=\{w\in \Sigma^*\mid \delta(s,w)\cap F\neq \emptyset\}$.

For the case of DFCAs, if $A$ is a DFCA for $L$ and two states are similar with respect
 to the similarity relation induced by $A$, then all the words reaching these states are similar.
Moreover, if two words of minimal length reach two distinct states in a DFCA, and the words are similar
with respect to $L$, then the states in the DFCA must be similar with respect to the similarity 
relation induced by $A$. These results are used  for DFCA minimization, and we need to verify 
if they can be used in case of NFCAs. In the following lemmata we show that  
 the corresponding results are no longer true.

\begin{lemma}
 Let $A=(Q,\Sigma,\delta,q_0,F)$ be a NFCA for the finite language $L$.
It is possible that $x_A(s)\sim_L x_A(q)$, but $s$ and $q$ are not mergible.
\end{lemma}
\begin{proof}
 For the automaton in Figure~\ref{f5}, left, $x_A(3)=x_A(1)$, but the states $1$ and $3$ are not mergible,
as the resulting automaton would not reject $a^7$.
\end{proof}

\begin{lemma}
 Let $A=(Q,\Sigma,\delta,q_0,F)$ be a NFCA for the finite language $L$, and $p,q\in Q$, $p\neq q$.
It is possible to have $x,y\in \Sigma^*$, $p\in \delta(q_0,x)$, $q\in \delta(q_0,y)$,
 $p \sim q$, and  $x\not\sim_L y$.
\end{lemma}
\begin{proof}
Consider the language $L=L(A)\cap\{a,b\}^{\leq14}$, where $A$ is depicted in Figure~\ref{f6}.
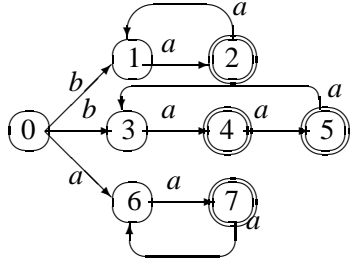
\begin{figure}
\begin{center}
\begin{picture}(100,100)
 \put(10,50){\oval(15,15)}
 \put(7,47){$0$}
 \put(16,50){\vector(1,1){25}}
 \put(25,65){$b$}
 \put(16,50){\vector(1,0){25}}
 \put(30,55){$b$}
 \put(16,50){\vector(1,-1){25}}
 \put(25,30){$a$}
 \put(49,77){\oval(15,15)}
 \put(47,74){$1$}
 \put(87,77){\oval(15,15)}
 \put(87,77){\oval(18,18)}
 \put(84,74){$2$}
 \put(67,88){\oval(40,20)[t]}
 \put(47,90){\vector(0,-1){5}}
 \put(87,90){\line(0,-1){5}}
 \put(87,94){$a$}

 \put(53,75){\vector(1,0){25}}
 \put(60,80){$a$}
 \put(47,50){\oval(15,15)}
 \put(45,47){$3$}
 \put(53,50){\vector(1,0){25}}
 \put(60,55){$a$}
 \put(85,50){\oval(15,15)}
 \put(85,50){\oval(18,18)}
 \put(82,47){$4$}
 \put(91,50){\vector(1,0){25}}
 \put(95,55){$a$}
 \put(122,50){\oval(15,15)}
 \put(122,50){\oval(18,18)}
 \put(120,47){$5$}
 \put(82,62){\oval(75,10)[t]}
 \put(45,63){\vector(0,-1){5}}
 \put(120,63){\line(0,-1){5}}
 \put(123,63){$a$}

 \put(49,23){\oval(15,15)}
 \put(47,20){$6$}
 \put(55,23){\vector(1,0){25}}
 \put(62,28){$a$}
 \put(87,23){\oval(15,15)}
 \put(87,23){\oval(18,18)}
 \put(84,20){$7$}
 \put(68,10){\oval(40,15)[b]}
 \put(48,10){\vector(0,1){5}}
 \put(88,10){\line(0,1){5}}
 \put(92,13){$a$}

\end{picture}
\end{center}
\begin{center}
\caption{{An example where $p\sim_A q$, $x\not\sim_L y$, but $p\in\delta(q_0,x)$ and $q\in\delta(q_0,y)$,
           $aa\not\sim_L ba$, 
$2\in \delta(0,ba)$, \newline$7\in \delta(0,aa)$, and $2\sim_A 7$.}
          }
\end{center}
 \label{f6}
\end{figure}

We have that:
\begin{itemize}
 \item $aa\not\sim_L ba$, because $aaa\notin L$, but $baa\in L$;
 \item $2\in \delta(0,ba)$, $7\in \delta(0,aa)$, and
 \item $2\sim_A 7$, because
       $\delta(2,a^{2k})=\{2\}\subseteq F$, $\delta(2,a^{2k+1})=\{1\}\cap F=\emptyset$,
       $\delta(7,a^{2k})=\{7\}\subseteq F$, $\delta(7,a^{2k+1})=\{6\}\cap F=\emptyset$,
     and $\delta(2,w)=\delta(7,w)=\emptyset$, for all $w\in\Sigma^*-\{a\}^*$.$\Box$
\end{itemize}
\end{proof}

Let us verify the case when two states $p,q$ are similar, or we can distinguish between them.

\begin{lemma}
  Let $A=(Q,\Sigma,\delta,q_0,F)$ be  a NFCA for the finite language $L$, $p,q\in Q$, $p\neq q$,
and either $p,q\in F$, or $p,q\notin F$.
Assume $r\in \delta(p,a)$ and $s\in \delta(q,a)$.
\begin{enumerate}
 \item  If $r\sim_A s $, for all possible choices of $r$ and $s$, then $p\sim_A q$.
 \item  The converse is false, i.e., we may have  $r\not\sim_A s $, for some $r$ and $s$, and $p\sim_A q$.
\end{enumerate}
\end{lemma}
\begin{proof}
Assume  $p\not\sim_A q$, and let $w\in \Sigma^{\leq l-\max\{level(p),level(q)\}}\cap\Sigma^+$. 
Because either $p,q\in F$, or $p,q\notin F$, 
we have that $\delta(p,aw)\cap F\neq \emptyset $ and $\delta(q,aw)\cap F= \emptyset $, 
or   $\delta(p,aw)\cap F= \emptyset $, and $\delta(q,aw)\cap F\neq \emptyset $.
If $\delta(p,aw)\cap F\neq \emptyset $ and $\delta(q,aw)\cap F= \emptyset $, 
it follows that we have two states $r\in \delta(p,a)$ and $s\in \delta(q,a)$
such that  $\delta(r,w)\cap F\neq \emptyset $, and $\delta(s,w)\cap F= \emptyset $.
This proves that the first implication is true. For the second implication, 
consider the automaton depicted in Figure~\ref{f6} with $l=14$, and the following states $p,q,r,s$:
 $p=q=0$, $r=1$, $s=3$, and the letter $b$. We have that 
$p\sim q$, $1,3\in\delta(p,b)=\delta(q,b)=\delta(0,b)$, 
but $r\not\sim s$, because $\delta(1,a)\cap F=\emptyset$ 
and $\delta(3,a)\cap F=\{4\}\neq \emptyset$.$\Box$
\end{proof}

This result contrasts with the one for the deterministic case for cover automata, 
and the main reason is the nondeterminism, not the fact that we work with  cover languages.

Next, we would like to verify if similar states can be merged in case of NFCAs, also to check
 which type of merge works.
In case we have two similar states, we can strongly merge them as shown below.
In the case of DFCAs, if two states are similar, these can be merged. We must ensure that 
the same result is also true for  NFCAs, and the next theorem 
shows it.

\begin{theorem}
\label{ltech1}
 Let $A=(Q,\Sigma,\delta,q_0,F)$ be an NFCA for $L$,
and $p,q\in Q$ such that $p\neq q$,  and $p\sim q$.
Then we have
\begin{enumerate}
 \item if $level_A(p)\leq level_A(q)$, then $p\precsim q$.
 \item It is possible that $p\not\precapprox q$.
\end{enumerate}
\end{theorem}

\begin{proof}
For the first part, let $A'$ be the automaton obtained from $A$ by strongly merging $q$ in $p$.
We need to show that
$A'$ is a cover NFCA for $L$.
Let $w=w_1\ldots w_n$ be a word in $\Sigma^{\leq l}$, $n\in \N$ and 
$w_i\in \Sigma$ for all $i$, $1\leq i\leq n$.
We now prove that $w \in L$ iff $\delta'(q_0, w) \cap F'\neq\emptyset$.

If we can find the states $\{q_0,q_1,\ldots,q_n\}$ such that
$q_1\in \delta(q_0,w_1)$,
$q_2\in \delta(q_1,w_2)$,
\ldots,
$q_n\in\delta(q_{n-1},w_n)$,
$q_n\in F$ and $q\notin \{q_0,q_1,\ldots,q_n\}$,
then 
$q_1\in \delta'(q_0,w_1)$,
$q_2\in \delta'(q_1,w_2)$,
\ldots,
$q_n\in \delta'(q_{n-1},w_n)$,
$q_n\in F'$, i.e., 
$\delta'(q_0, w) \cap F'\neq \emptyset$.
Assume $q=q_j$, and $j$ is the smallest with this property.
If $j=n$, then $q\in F$, which implies $p\in F$, then 
$q_1\in \delta'(q_0,w_1)$,
$q_2\in \delta'(q_1,w_2)$,
\ldots,
$q_n\in\delta'(p,w_n)$,
which means  $\delta'(q_0, w) \cap F'\neq\emptyset$.

Assume the statements hold for $|w_j\ldots w_n|<l'$ for $l'<l-|w|$
($l-|w_1...w_j|\leq l-level(q)$), and 
consider the case when 
 $|w_{j-1}w_{j}\ldots w_n|=l'$.
If for every non-empty prefix of $w_{j+1}\ldots w_n$, $w_{j-1}\ldots w_h$,
$q\notin \delta(p,w_{j-1}\ldots w_h)$, then
$\delta(p,w_{j+1}\ldots w_n)\in F-\{q\}$ iff 
$\delta(q,w_{j+1}\ldots w_n)\in F-\{q\}$ , i.e., 
$\delta'(p,w_{j+1}\ldots w_n)\cap F'\neq \emptyset$
iff
$\delta(q,w_{j+1}\ldots w_n)\cap F\neq \emptyset$.

Otherwise, let $h$ be the smallest number such that 
$q\in \delta(q,w_{j+1}\ldots w_h$.
Then $|w_{h+1}\ldots w_n|<l'$ (and $p\in \delta'(p,w_j\ldots w_h)$).
By induction hypothesis,
$\delta'(p,w_{h+1}\ldots w_n)\cap F'\neq \emptyset$ iff $\delta(q,w_{h+1}\ldots w_n)\cap F\neq \emptyset$.
Therefore, $\delta(p,w_{j+1}\ldots w_hw_{h+1}\ldots w_n)\cap F'\neq \emptyset$ iff 
$\delta(q, w_{j+1}\ldots w_hw_{h+1}\ldots w_n)\cap F\neq\emptyset$,
proving the first part.
For the second part, consider the automaton in Figure~\ref{f7} as a NFCA for $L=\{a^2,a^4\}$. We have that
 $l=4$ and $3\sim 5$, because $level(3)=3$, and $\delta(3,\varepsilon)\cap F=\delta(5,\varepsilon)\cap F=\emptyset$
$\delta(3,a)\cap F=\{4\}$, $\delta(5,a)\cap F=\{6\}$.
We cannot weakly merge state $3$ with state $5$, as we would recognize $a^3\notin L$.
In Figure~\ref{f8} we have the result for strongly merging state $3$ in state $5$.
\begin{figure}
\begin{picture}(150,80)
\put(-5,50){\vector(1,0){5}}  
\put(8,50){\oval(15,15)}  
\put(5,47){$0$}
\put(16,50){\vector(1,0){27}}
\put(25,55){$a$}    
\put(16,48){\vector(1,-1){26}}
\put(20,30){$a$}
\put(50,20){\oval(15,15)}
\put(47,17){$5$}    
\put(57,20){\vector(1,0){24}}
\put(65,25){$a$}
\put(90,20){\oval(15,15)}
\put(90,20){\oval(18,18)}
\put(87,17){$6$}    

\put(50,50){\oval(15,15)}  
\put(47,47){$1$}
\put(57.5,50){\vector(1,0){25}}  
\put(65,55){$a$}    

\put(90,50){\oval(15,15)}  
\put(87,47){$2$}
\put(98,50){\vector(1,0){25}}  
\put(105,55){$a$}    

\put(130,50){\oval(15,15)}  
\put(127,47){$3$}
\put(138,50){\vector(1,0){25}}  
\put(145,55){$a$}    
\put(130,71){\oval(15,15)[t]}  
\put(137,71){\line(-1,-3){5}}  
\put(123,71){\vector(1,-3){5}}  
\put(137,77){$a$}    

\put(171,50){\oval(15,15)}
\put(171,50){\oval(18,18)}    
\put(168,47){$4$}
\end{picture}
\hfill
\begin{picture}(150,80)
\put(-7,60){\vector(1,0){7}}
\put(8,60){\oval(15,15)}  
\put(5,57){$0$}
\put(15,60){\vector(1,0){27}}
\put(25,65){$a$}    
\put(15,57){\vector(1,-1){23}}
\put(18,43){$a$}
\put(47,30){\oval(19,15)}
\put(39.5,27){$3,5$}    
\put(37.5,30){\vector(-1,0){21}}  
\put(56.5,30){\vector(1,0){27}}
\put(65,33){$a$}
\put(92,30){\oval(15,15)}
\put(92,30){\oval(18,18)}
\put(89,27){$6$}    

\put(50,60){\oval(15,15)}  
\put(47,57){$1$}
\put(58,60){\vector(1,0){26.5}}  
\put(65,64){$a$}    

\put(92,60){\oval(15,15)}  
\put(89,57){$2$}
\put(92,52){\vector(-2,-1){36}}  
\put(65,45){$a$}    


\put(20,35){$a$}    
\put(47,7){\oval(15,15)[b]}  
\put(55,7){\line(-1,3){5}}  
\put(40,7){\vector(1,3){5}}  
\put(57,5){$a$}    

\put(8,30){\oval(15,15)}
\put(8,30){\oval(18,18)}    
\put(5,27){$4$}
\end{picture}
\caption{Example for weakly merging failure and similar states.}
\label{f7}
\end{figure}
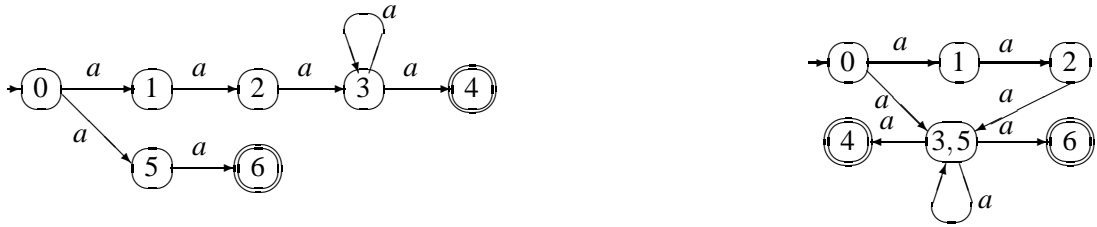
\end{proof}

\begin{figure}
\begin{center}
\begin{picture}(100,60)
\put(-5,50){\vector(1,0){5}}
\put(8,50){\oval(15,15)}  
\put(5,47){$0$}
\put(15.5,50){\vector(1,0){26.5}}
\put(25,55){$a$}    
\put(15.5,48){\vector(1,-1){22.5}}
\put(20,30){$a$}
\put(46,20){\oval(18,15)}
\put(43,17){$5$}    
\put(56,20){\vector(1,0){28}}
\put(65,22){$a$}
\put(92,20){\oval(15,15)}
\put(92,20){\oval(18,18)}
\put(89,17){$6$}    

\put(50,50){\oval(15,15)}  
\put(47,47){$1$}
\put(58,50){\vector(1,0){26.5}}  
\put(65,55){$a$}    

\put(92,50){\oval(15,15)}  
\put(89,47){$2$}
\put(90,42){\vector(-2,-1){35}}  
\put(65,33){$a$}    



\end{picture}
\end{center}
\caption{Example for strongly merging similar states for the example presented in Figure~\ref{f7}.}
\label{f8}
\end{figure}
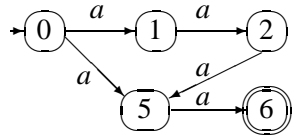
We can observe that strongly merging states does not add words in the language, 
while weakly merging may add words. 
Because any DFCA is also a NFCA, then some smaller automata can be obtained from 
larger ones without using state merging technique, and the following lemma presents such a case. 
Also, the automaton in Figure~\ref{f2} is obtained from automaton in Figure~\ref{f1} by strongly merging states 
$0,\ldots -m+1$ into state $-m$.

\begin{lemma}
Let $A=(Q,\Sigma,\delta,q_0,F)$ be an NFCA for $L$, and consider the reduced sub-automaton
generated by state $p$, $A=(Q_R,\Sigma,\delta_R,p,F)$, i.e.,  $Q_R$ 
contains only reachable and useful states, and $\delta_R$ is the induced transition function.
If $\delta(s,a)\cap Q_R=\emptyset$, for all $s\in (Q\setminus Q_R)$, 
 we can find two regular languages $L_1$,$L_2$ such that
\begin{itemize}
 \item $L_p=(L_1\cup L_2)\cap \Sigma^{\leq l-level(p)}$, and
 \item $nsc(L_1)+nsc(L_2)<\#Q_R+1$,
\end{itemize}
then $A$ is not minimal.
\end{lemma}
\begin{proof}
Let $A_i=(Q_i,\Sigma,\delta_i,q_{0,i},F_i)$, $i=1,2 $ be two NFAs
for $L_1$ and $L_2$, and  $L_p=(L_1\cup L_2)\cap \Sigma^{\leq l-level(p)}$.
We define the automaton $B=((Q\setminus Q_R)\cup \{p\}\cup Q_1\cup Q_2, \Sigma, \delta_B,q_0,F_B)$ as follows:
$F=(F\setminus Q_R)\cup F_1\cup F_2$, in case $p\notin F$, 
and $F=(F\setminus Q_R)\cup F_1\cup F_2\cup\{p\}$ in case $p\in F$.
For the transition function, we have $\delta_B(s,a)=\delta(s,a)$ if $s\in (Q\setminus Q_R)$,
$\delta_B(s,a)=\delta_i(s,a)$ if $s\in Q_i$, $i=1,2$,
and
$\delta_B(p,a)=\delta_1(q_{0,1},a)\cup\delta_2(q_{0,2},a)\cup\delta(p,a)\setminus Q_R$, if $p\notin\delta(p,a)$,
and $\delta_B(p,a)=\delta_1(q_{0,1},a)\cup\delta_2(q_{0,2},a)\cup\delta(p,a)\setminus Q_R\cup\{p\}$, if $p\in\delta(p,a)$.
Obviously, the automaton $B$ recognizes the cover language for $L$, 
and its state complexity is lower.
\end{proof}
This technique was used to produce the minimal NFCA for $L_{9,7}$ in Figure~\ref{f5}.

\section{Conclusion}
\label{openproblems}
In this paper we showed that NFCAs are a more compact representation of finite languages than both 
NFAs and DFCAs, therefore it is a subject worth investigating.
We presented a lower-bound technique for state complexity of NFCAs, and proved its limitations.
We showed that minimizing NFCAs has at least the same level of difficulty as minimizing 
general NFAs, and that extra information about the maximum length of the words in the language 
does not help reducing the time complexity.
We checked if some of the results involving reducing the size of automata for NFAs and DFCAs
are still valid for NFCAs, and showed that most of them are no longer valid. 
However, the method of strong merging states still works in case of NFCAs, and we showed 
that there are also other methods that could be investigated.

As future research, below is a list of problems we consider worth investigating:
\begin{enumerate}
 \item check if the bipartite graph lower-bound
 technique can be applied for \mbox{NFCAs;}
 \item find bounds for nondeterministic cover state complexity;
 \item investigate the problem of magic numbers for NFCAs. 
In this case, we can relate either to DFCAs, or DFAs. 
\end{enumerate}

\bibliographystyle{eptcs.bst}
\bibliography{MinNFCA}

\begin{thebibliography}{10}
\providecommand{\bibitemdeclare}[2]{}
\providecommand{\surnamestart}{}
\providecommand{\surnameend}{}
\providecommand{\urlprefix}{Available at }
\providecommand{\url}[1]{\texttt{#1}}
\providecommand{\href}[2]{\texttt{#2}}
\providecommand{\urlalt}[2]{\href{#1}{#2}}
\providecommand{\doi}[1]{doi:\urlalt{http://dx.doi.org/#1}{#1}}
\providecommand{\bibinfo}[2]{#2}

\bibitemdeclare{article}{Primtest}
\bibitem{Primtest}
\bibinfo{author}{M.~\surnamestart Agrawal\surnameend},
  \bibinfo{author}{N.~\surnamestart Kayal\surnameend} \&
  \bibinfo{author}{N.~\surnamestart Saxena\surnameend} (\bibinfo{year}{2004}):
  \emph{\bibinfo{title}{PRIMES is in P.}}
\newblock {\sl \bibinfo{journal}{Annals of mathematics}}, pp.
  \bibinfo{pages}{781--793}.
\newblock \urlprefix\url{http://dx.doi.org/10.4007/annals.2004.160.781}.

\bibitemdeclare{article}{heuristic}
\bibitem{heuristic}
\bibinfo{author}{J.~\surnamestart Amilhastre\surnameend},
  \bibinfo{author}{P.~\surnamestart Janssen\surnameend} \&
  \bibinfo{author}{M-C. \surnamestart Vilarem\surnameend}
  (\bibinfo{year}{2001}): \emph{\bibinfo{title}{FA Minimisation Heuristics for
  a Class of Finite Languages}}.
\newblock {\sl \bibinfo{journal}{Lecture Notes in Computer Science}}
  \bibinfo{volume}{2214}, pp. \bibinfo{pages}{1 -- 12}.
\newblock \urlprefix\url{http://dx.doi.org/10.1007/3-540-45526-4_1}.

\bibitemdeclare{article}{birget}
\bibitem{birget}
\bibinfo{author}{J.-C. \surnamestart Birget\surnameend} (\bibinfo{year}{1992}):
  \emph{\bibinfo{title}{Intersection and union of regular languages and state
  complexity}}.
\newblock {\sl \bibinfo{journal}{Information Processing Letters}}
  \bibinfo{volume}{43}, pp. \bibinfo{pages}{185--190}.
\newblock \urlprefix\url{http://dx.doi.org/10.1016/0020-0190(92)90198-5}.

\bibitemdeclare{article}{gapIJFCS}
\bibitem{gapIJFCS}
\bibinfo{author}{C.~\surnamestart C\^ampeanu\surnameend},
  \bibinfo{author}{A.~P\u \surnamestart aun\surnameend} \&
  \bibinfo{author}{S.~\surnamestart Yu\surnameend} (\bibinfo{year}{2002}):
  \emph{\bibinfo{title}{An Efficient Algorithm for Constructing Minimal Cover
  Automata for Finite Languages}}.
\newblock {\sl \bibinfo{journal}{International Journal of Foundations of
  Computer Science}} \bibinfo{volume}{13}(\bibinfo{number}{1}), pp.
  \bibinfo{pages}{83 -- 97}.
\newblock \urlprefix\url{http://dx.doi.org/10.1142/S0129054102000960}.

\bibitemdeclare{article}{finiteop}
\bibitem{finiteop}
\bibinfo{author}{C.~\surnamestart C\^ampeanu\surnameend},
  \bibinfo{author}{K.~Culik \surnamestart II\surnameend},
  \bibinfo{author}{K.~\surnamestart Salomaa\surnameend} \&
  \bibinfo{author}{S.~\surnamestart Yu\surnameend} (\bibinfo{year}{2001}):
  \emph{\bibinfo{title}{State complexity of basic operations on finite
  languages}}.
\newblock {\sl \bibinfo{journal}{Lecture Notes in Computer Science}}
  \bibinfo{volume}{2214}, pp. \bibinfo{pages}{60--70}.
\newblock \urlprefix\url{http://dx.doi.org/10.1007/3-540-45526-4_6}.

\bibitemdeclare{article}{CoverAutomata}
\bibitem{CoverAutomata}
\bibinfo{author}{C.~\surnamestart C\^ampeanu\surnameend},
  \bibinfo{author}{N.~\surnamestart Santean\surnameend} \&
  \bibinfo{author}{S.~\surnamestart Yu\surnameend} (\bibinfo{year}{1986}):
  \emph{\bibinfo{title}{Minimal cover-automata for finite languages}}.
\newblock {\sl \bibinfo{journal}{Theoretical Computer Science}}
  \bibinfo{volume}{267}(\bibinfo{number}{1-2}), pp. \bibinfo{pages}{3--16}.
\newblock \urlprefix\url{http://dx.doi.org/10.1016/S0304-3975(00)00292-9}.

\bibitemdeclare{article}{Chrobak}
\bibitem{Chrobak}
\bibinfo{author}{M.~\surnamestart Chrobak\surnameend} (\bibinfo{year}{1986}):
  \emph{\bibinfo{title}{Finite Automata and Unary Languages}}.
\newblock {\sl \bibinfo{journal}{Theoretical Computer Science}}
  \bibinfo{volume}{47}(\bibinfo{number}{2}), pp. \bibinfo{pages}{149--158}.
\newblock \urlprefix\url{http://dx.doi.org/10.1016/0304-3975(86)90142-8}.

\bibitemdeclare{article}{gramlich}
\bibitem{gramlich}
\bibinfo{author}{\surnamestart G.Gramlich\surnameend} (\bibinfo{year}{2003}):
  \emph{\bibinfo{title}{Probabilistic and Nondeterministic Unary Automata}}.
\newblock {\sl \bibinfo{journal}{Lecture Notes in Computer Science}}
  \bibinfo{volume}{2747}, pp. \bibinfo{pages}{460 -- 469}.
\newblock \urlprefix\url{http://dx.doi.org/10.1007/978-3-540-45138-9_40}.

\bibitemdeclare{article}{Shallit}
\bibitem{Shallit}
\bibinfo{author}{I.~\surnamestart Glaister\surnameend} \&
  \bibinfo{author}{J.~\surnamestart Shallit\surnameend} (\bibinfo{year}{1996}):
  \emph{\bibinfo{title}{A lower bound technique for the size of
  nondeterministic finite automata}}.
\newblock {\sl \bibinfo{journal}{Information Processing Letters}}
  \bibinfo{volume}{59}, pp. \bibinfo{pages}{75 -- 77}.
\newblock \urlprefix\url{http://dx.doi.org/10.1016/0020-0190(96)00095-6}.

\bibitemdeclare{article}{GruberHolzerNFAHard}
\bibitem{GruberHolzerNFAHard}
\bibinfo{author}{H.~\surnamestart Gruber\surnameend} \&
  \bibinfo{author}{M.~\surnamestart Holzer\surnameend} (\bibinfo{year}{2006}):
  \emph{\bibinfo{title}{Finding lower bounds for nondeterministic state
  complexity is hard}}.
\newblock {\sl \bibinfo{journal}{Lecture Notes in Computer Science}}
  \bibinfo{volume}{4036}, pp. \bibinfo{pages}{363--374}.
\newblock \urlprefix\url{http://dx.doi.org/10.1007/11779148_33}.

\bibitemdeclare{article}{gruberholzerunary}
\bibitem{gruberholzerunary}
\bibinfo{author}{H.~\surnamestart Gruber\surnameend} \&
  \bibinfo{author}{M.~\surnamestart Holzer\surnameend} (\bibinfo{year}{2007}):
  \emph{\bibinfo{title}{Computational Complexity of NFA Minimization for Finite
  and Unary Languages}}.
\newblock {\sl \bibinfo{journal}{LATA}} \bibinfo{volume}{8}, pp.
  \bibinfo{pages}{261--272}.
\newblock
  \urlprefix\url{http://www2.tcs.ifi.lmu.de/~gruberh/data/lata07-submission.pd%
f}.

\bibitemdeclare{article}{holzerKutribNFA}
\bibitem{holzerKutribNFA}
\bibinfo{author}{M.~\surnamestart Holzer\surnameend} \&
  \bibinfo{author}{M.~\surnamestart Kutrib\surnameend} (\bibinfo{year}{2003}):
  \emph{\bibinfo{title}{State complexity of basic operations on
  nondeterministic finite automata}}.
\newblock {\sl \bibinfo{journal}{Lecture Notes in Computer Science}}
  \bibinfo{volume}{2608}, pp. \bibinfo{pages}{148--157}.
\newblock \urlprefix\url{http://dx.doi.org/10.1007/3-540-44977-9_14}.

\bibitemdeclare{article}{holzerKutribUnary}
\bibitem{holzerKutribUnary}
\bibinfo{author}{M.~\surnamestart Holzer\surnameend} \&
  \bibinfo{author}{M.~\surnamestart Kutrib\surnameend} (\bibinfo{year}{2003}):
  \emph{\bibinfo{title}{Unary language operations and their nondeterministic
  state complexity}}.
\newblock {\sl \bibinfo{journal}{Lecture Notes in Computer Science}}
  \bibinfo{volume}{2450}, pp. \bibinfo{pages}{162--172}.
\newblock \urlprefix\url{http://dx.doi.org/10.1007/3-540-45005-X_14}.

\bibitemdeclare{article}{holzerKutribLata09}
\bibitem{holzerKutribLata09}
\bibinfo{author}{M.~\surnamestart Holzer\surnameend} \&
  \bibinfo{author}{M.~\surnamestart Kutrib\surnameend} (\bibinfo{year}{2009}):
  \emph{\bibinfo{title}{Descriptional and computational complexity of finite
  automata}}.
\newblock {\sl \bibinfo{journal}{Lecture Notes in Computer Science}}
  \bibinfo{volume}{5457}, pp. \bibinfo{pages}{23--42}.
\newblock \urlprefix\url{http://dx.doi.org/10.1007/978-3-642-00982-2_3}.

\bibitemdeclare{article}{holzerKutribNFA09}
\bibitem{holzerKutribNFA09}
\bibinfo{author}{M.~\surnamestart Holzer\surnameend} \&
  \bibinfo{author}{M.~\surnamestart Kutrib\surnameend} (\bibinfo{year}{2009}):
  \emph{\bibinfo{title}{Nondeterministic finite automata - recent results on
  the descriptional and computational complexity}}.
\newblock {\sl \bibinfo{journal}{Int. J. Found. Comput. Sci}}
  \bibinfo{volume}{20}(\bibinfo{number}{4}), pp. \bibinfo{pages}{563--580}.
\newblock \urlprefix\url{http://dx.doi.org/10.1142/S0129054109006747}.

\bibitemdeclare{incollection}{hopcroft}
\bibitem{hopcroft}
\bibinfo{author}{John \surnamestart Hopcroft\surnameend}
  (\bibinfo{year}{1971}): \emph{\bibinfo{title}{An $n \mathop {\mathgroup
  \symoperators log}\nolimits n$ Algorithm for Minimizing States in a Finite
  Automaton}}.
\newblock In \bibinfo{editor}{Z.~\surnamestart Kohavi\surnameend} \&
  \bibinfo{editor}{A.~\surnamestart Paz\surnameend}, editors: {\sl
  \bibinfo{booktitle}{Theory of Machines and Computations}},
  \bibinfo{publisher}{Academic Press, New York}, pp. \bibinfo{pages}{189--196}.

\bibitemdeclare{book}{hopcroft79}
\bibitem{hopcroft79}
\bibinfo{author}{John~E. \surnamestart Hopcroft\surnameend} \&
  \bibinfo{author}{Jeffrey~D. \surnamestart Ullman\surnameend}
  (\bibinfo{year}{1979}): \emph{\bibinfo{title}{Introduction to Automata
  Theory, Languages and Computation}}.
\newblock \bibinfo{publisher}{Addison-Wesley}.

\bibitemdeclare{article}{ilieYunfa}
\bibitem{ilieYunfa}
\bibinfo{author}{L.~\surnamestart Ilie\surnameend},
  \bibinfo{author}{G.~\surnamestart Navarro\surnameend} \&
  \bibinfo{author}{S.~\surnamestart Yu\surnameend} (\bibinfo{year}{2004}):
  \emph{\bibinfo{title}{On NFA reductions}}.
\newblock {\sl \bibinfo{journal}{Lecture Notes in Computer Science Volume:
  Theory Is Forever Essays Dedicated to Arto Salomaa on the Occasion of His
  70th Birthday}} \bibinfo{volume}{3113}, pp. \bibinfo{pages}{112--124}.
\newblock \urlprefix\url{http://dx.doi.org/10.1007/978-3-540-27812-2_11}.

\bibitemdeclare{article}{ravikumar}
\bibitem{ravikumar}
\bibinfo{author}{T.~\surnamestart Jiang\surnameend} \&
  \bibinfo{author}{B.~\surnamestart Ravikumar\surnameend}
  (\bibinfo{year}{1993}): \emph{\bibinfo{title}{NFA minimization problems are
  hard}}.
\newblock {\sl \bibinfo{journal}{SIAM Journal on Computing}}
  \bibinfo{volume}{22}(\bibinfo{number}{1}), pp. \bibinfo{pages}{117--141}.

\bibitemdeclare{book}{china}
\bibitem{china}
\bibinfo{author}{N.~\surnamestart Koblitz\surnameend} (\bibinfo{year}{1994}):
  \emph{\bibinfo{title}{A Course in Number Theory and Criptography}}.
\newblock \bibinfo{publisher}{Springer}.
\newblock \urlprefix\url{http://dx.doi.org/10.1007/978-1-4419-8592-7}.

\bibitemdeclare{article}{KornerGoeman}
\bibitem{KornerGoeman}
\bibinfo{author}{H.~\surnamestart Körner\surnameend} (\bibinfo{year}{2003}):
  \emph{\bibinfo{title}{A Time and Space Efficient Algorithm for Minimizing
  Cover Automata for Finite Languages}}.
\newblock {\sl \bibinfo{journal}{International Journal of Foundations of
  Computer Science}} \bibinfo{volume}{14}(\bibinfo{number}{6}), pp.
  \bibinfo{pages}{1071--1086}.
\newblock \urlprefix\url{http://dx.doi.org/10.1142/S0129054103002187}.

\bibitemdeclare{article}{maslov}
\bibitem{maslov}
\bibinfo{author}{A.~N. \surnamestart Maslov\surnameend} (\bibinfo{year}{1970}):
  \emph{\bibinfo{title}{Estimates of the number of states of finite automata}}.
\newblock {\sl \bibinfo{journal}{Soviet Mathematics Doklady}}
  \bibinfo{volume}{11}, pp. \bibinfo{pages}{1373--1374}.

\bibitemdeclare{article}{maslov1}
\bibitem{maslov1}
\bibinfo{author}{A.~N. \surnamestart Maslov\surnameend} (\bibinfo{year}{1973}):
  \emph{\bibinfo{title}{Cyclic shift operation for languages}}.
\newblock {\sl \bibinfo{journal}{Probl. Inf. Transm}} \bibinfo{volume}{9}, pp.
  \bibinfo{pages}{333--338}.

\bibitemdeclare{article}{pighizzini}
\bibitem{pighizzini}
\bibinfo{author}{F.~\surnamestart Mera\surnameend} \&
  \bibinfo{author}{G.~\surnamestart Pighizzini\surnameend}
  (\bibinfo{year}{2005}): \emph{\bibinfo{title}{Complementing unary
  nondeterministic automata}}.
\newblock {\sl \bibinfo{journal}{Theoretical Computer Science}}
  \bibinfo{volume}{330}, pp. \bibinfo{pages}{349--360}.
\newblock \urlprefix\url{http://dx.doi.org/10.1016/j.tcs.2004.04.015}.

\bibitemdeclare{article}{moore}
\bibitem{moore}
\bibinfo{author}{E.~F. \surnamestart Moore\surnameend} (\bibinfo{year}{1956}):
  \emph{\bibinfo{title}{Gedanken-experiments on sequential machines}}.
\newblock {\sl \bibinfo{journal}{Automata studies, Annals of mathematics
  studies}} \bibinfo{volume}{34}, pp. \bibinfo{pages}{129--153}.

\bibitemdeclare{article}{revuz}
\bibitem{revuz}
\bibinfo{author}{D.~\surnamestart Revuz\surnameend} (\bibinfo{year}{1992}):
  \emph{\bibinfo{title}{Minimisation of acyclic deterministic automata in
  linear time}}.
\newblock {\sl \bibinfo{journal}{Theoretical Computer Science}}
  \bibinfo{volume}{92}(\bibinfo{number}{1}), pp. \bibinfo{pages}{181 -- 189}.
\newblock \urlprefix\url{http://dx.doi.org/10.1147/rd.32.0114}.

\bibitemdeclare{article}{KaiShengsc92}
\bibitem{KaiShengsc92}
\bibinfo{author}{S.~\surnamestart Yu\surnameend},
  \bibinfo{author}{K.~\surnamestart Salomaa\surnameend} \&
  \bibinfo{author}{Q.~\surnamestart Zhuang\surnameend} (\bibinfo{year}{1994}):
  \emph{\bibinfo{title}{The state complexities of some basic operations on
  regular languages}}.
\newblock {\sl \bibinfo{journal}{Theoretical Computer Science}}
  \bibinfo{volume}{125}(\bibinfo{number}{2}), pp. \bibinfo{pages}{315--328}.
\newblock \urlprefix\url{http://dx.doi.org/10.1016/0304-3975(92)00011-F}.

\end{thebibliography}
\end{document}